\documentclass[prl,superscriptaddress,twocolumn]{revtex4}

\usepackage{amsfonts,amssymb,amsmath}
\usepackage[]{graphics,graphicx,epsfig}
\usepackage{amsthm}

\usepackage{mathtools}
\newcommand{\defeq}{\vcentcolon=}

\def\identity{\leavevmode\hbox{\small1\kern-3.8pt\normalsize1}}

\newcommand{\half}{\mbox{$\textstyle \frac{1}{2}$}}

\renewcommand{\epsilon}{\varepsilon}

\newtheorem{definition}{Definition} %[section]
\newtheorem{definition2}{Definition}

\newtheorem{thm}[definition]{Theorem}
\newtheorem{corollary}[definition]{Corollary}

\newtheorem*{rep@theorem}{\rep@title}
\newcommand{\newreptheorem}[2]{%
\newenvironment{rep#1}[1]{%
 \def\rep@title{#2 \ref{##1} (restatement)}%
 \begin{rep@theorem}}%
 {\end{rep@theorem}}}
\makeatother

\newreptheorem{thm}{Theorem}
\newreptheorem{lem}{Lemma}

\def\ba#1\ea{\begin{align}#1\end{align}}
\def\ban#1\ean{\begin{align*}#1\end{align*}}

\newcommand{\be}{\begin{equation}}
\newcommand{\ee}{\end{equation}}

\def\benum{\begin{enumerate}}
\def\eenum{\end{enumerate}}

\def\squareforqed{\hbox{\rlap{$\sqcap$}$\sqcup$}}
\def\qed{\ifmmode\squareforqed\else{\unskip\nobreak\hfil
\penalty50\hskip1em\null\nobreak\hfil\squareforqed
\parfillskip=0pt\finalhyphendemerits=0\endgraf}\fi}
\def\endenv{\ifmmode\;\else{\unskip\nobreak\hfil
\penalty50\hskip1em\null\nobreak\hfil\;
\parfillskip=0pt\finalhyphendemerits=0\endgraf}\fi}
\newcommand{\bra}[1]{\langle #1|}
\newcommand{\ket}[1]{|#1\rangle}
\newcommand{\braket}[2]{\langle #1|#2\rangle}

\newcommand{\<}{\langle}
\renewcommand{\>}{\rangle}

\def\be{\begin{equation}}
\def\ee{\end{equation}}
\def\ben{\begin{eqnarray}}
\def\een{\end{eqnarray}}

\def\bei{\begin{itemize}}
\def\eei{\end{itemize}}

\mathchardef\ordinarycolon\mathcode`\:
\mathcode`\:=\string"8000
\def\vcentcolon{\mathrel{\mathop\ordinarycolon}}
\begingroup \catcode`\:=\active
  \lowercase{\endgroup
  \let :\vcentcolon
  }

\newcommand{\nc}{\newcommand}
 \nc{\proj}[1]{|#1\rangle\!\langle #1 |} 
\nc{\avg}[1]{\langle#1\rangle}

\nc{\conv}{\operatorname{conv}}
\nc{\smfrac}[2]{\mbox{$\frac{#1}{#2}$}} \nc{\Tr}{\operatorname{Tr}}
\nc{\ox}{\otimes} \nc{\dg}{\dagger} \nc{\dn}{\downarrow}
\nc{\lmax}{\lambda_{\text{max}}}
\nc{\lmin}{\lambda_{\text{min}}}

\nc{\csupp}{{\operatorname{csupp}}}
\nc{\qsupp}{{\operatorname{qsupp}}} \nc{\var}{\operatorname{var}}
\nc{\rar}{\rightarrow} \nc{\lrar}{\longrightarrow}
\nc{\poly}{\operatorname{poly}}
\nc{\polylog}{\operatorname{polylog}} \nc{\Lip}{\operatorname{Lip}}
\nc{\Om}{\Omega}
\nc{\wt}[1]{\widetilde{#1}}

\def\>{\rangle}
\def\<{\langle}

\nc{\glneq}{{\raisebox{0.6ex}{$>$}  \hspace*{-1.8ex} \raisebox{-0.6ex}{$<$}}}
\nc{\gleq}{{\raisebox{0.6ex}{$\geq$}\hspace*{-1.8ex} \raisebox{-0.6ex}{$\leq$}}}

\nc{\vholder}[1]{\rule{0pt}{#1}}
\nc{\wh}[1]{\widehat{#1}}
\nc{\h}[1]{\widehat{#1}}

\nc{\ob}[1]{#1}

\def\beq{\begin {equation}}
\def\eeq{\end {equation}}

\def\be{\begin{equation}}
\def\ee{\end{equation}}

\nc{\eq}[1]{(\ref{eq:#1})} 
\nc{\eqs}[2]{\eq{#1} and \eq{#2}}

\nc{\eqn}[1]{Eq.~(\ref{eqn:#1})}
\nc{\eqns}[2]{Eqs.~(\ref{eqn:#1}) and (\ref{eqn:#2})}

\nc{\region}{\cS\cW}
\newenvironment{protocol*}[1]
  {
    \begin{center}
      \hrulefill\\
      \textbf{#1}
  }
  {
    \vspace{-1\baselineskip}
    \hrulefill
    \end{center}
  }

\begin{document}

\title{Characterising the Performance of XOR Games and the Shannon Capacity of Graphs}
\author{Ravishankar \surname{Ramanathan}}
\email{ravishankar.r.10@gmail.com}
\affiliation{National Quantum Information Center of Gda\'nsk,  81-824 Sopot, Poland}
\affiliation{University of Gda\'nsk, 80-952 Gda\'nsk, Poland}
\author{Alastair Kay}
\affiliation{Department of Mathematics, Royal Holloway University of London, Egham, Surrey, TW20 0EX, UK}
\author{Gl\'aucia Murta}
\affiliation{National Quantum Information Center of Gda\'nsk, 81-824 Sopot, Poland}
\affiliation{Departamento de Fisica, Universidade Federal de Minas Gerais, Caixa Postal 702, 30123-970, Belo Horizonte, MG, Brazil}
\author{Pawe{\l} \surname{Horodecki}}
\affiliation{National Quantum Information Center of Gda\'nsk, 81-824 Sopot, Poland}
\affiliation{Faculty of Applied Physics and Mathematics, Technical University of Gda\'nsk, 80-233 Gda\'nsk, Poland}

\begin{abstract}
In this paper we give a set of necessary and sufficient conditions such that quantum players of a two-party {\sc xor} game cannot perform any better than classical players. With any such game, we associate a graph and examine its zero-error communication capacity. This allows us to specify a broad new class of graphs for which the Shannon capacity can be calculated. The conditions also enable the parametrisation of new families of games which have no quantum advantage, for arbitrary input probability distributions up to certain symmetries. In the future, these might be used in information-theoretic studies on reproducing the set of quantum non-local correlations.
\end{abstract}
\maketitle

{\it Introduction:} A non-local game is one in which several distantly separated players are asked questions by a referee. Upon receipt of their answers, the referee computes whether or not they won, recording their success rate in a variable $\omega$. Although the players are not allowed to communicate, they can pre-share resources which may be consumed during the game.

Studying the contrast in the success when the shared resources are constrained to be quantum or classical is central to quantum information theory. The quintessential example is the Bell test \cite{bell,chsh}, in which there is an advantage to possessing quantum resources, and experiments indeed show this to be the case \cite{aspect}. These Bell tests have led to numerous philosophical and practical developments, such as demonstrating the indeterminism of Nature \cite{Pironio}, and providing the primary technical tool of quantum cryptography \cite{Ekert}. However, it is not only games for which there is a quantum advantage that are of interest; those for which there is no quantum advantage reveal just as much about Nature \cite{Winter2}.

One such class of games are the non-local computation (NLC) games \cite{NLC}, for which there is no quantum advantage for any probability distribution constrained to conform to a particular pattern. This is such a strong statement that it can be taken as an information-theoretic principle similar to those of \cite{IC, CC, LO, ML, Cabello3}: how would a world in which NLC games have no advantage over classical behave, and to what extent is quantum mechanics reproduced? To further such aims, an accurate characterisation of the games for which there is no quantum advantage is essential.

In this paper, we study {\sc xor} games, in which two parties Alice and Bob are asked questions $x,y\in[m]$ respectively, each giving a single bit answer, $a$ or $b$. The referee decides if they've won based on comparing $a\oplus b$ to $f(x,y)$, a deterministic function based on the two questions. We give necessary and sufficient conditions such that the classical and quantum scores are equal, $\omega_c=\omega_q$, subsequently describing broad families of examples.

It has recently been shown that any non-local game can be associated with a graph, $G$, for which some of the graph properties are closely associated with the quantum and classical behaviours of the game \cite{Cabello, Cabello2, Winter, LO, Fritz, Chailloux}:
\begin{eqnarray*}
&m^2\omega_c&=\alpha(G),	\\
\alpha(G)\leq& m^2\omega_q&\leq\theta(G)
\end{eqnarray*}
where $\alpha(G)$ is the independence number of $G$, i.e.\ the maximum number of mutually non-adjacent vertices, and $\theta(G)$ is known as the Lov\'asz theta function \cite{Lovasz-0}.

The zero-error capacity $\Theta(G)$ for sequential uses of a memoryless channel is the maximum rate at which information can be sent through the channel with zero probability of error, and this quantity is traditionally described using the confusability graph $G$ of the channel \cite{Shannon}. The vertices of the graph correspond to the inputs of the channel (letters of the encoding alphabet) and two vertices are connected by an edge if the corresponding inputs can be confused with each other by the receiver after transmission through the channel. The maximum number of one-letter messages which can be sent without confusion is then given by the independence number of the graph $\alpha(G)$. Denoting by $\alpha(G^k)$ the maximum number of $k$-letter messages that can be sent without confusion (two distinct $k$-letter words are confusable if, for every letter in one word, it is either confusable or equal to the corresponding letter in the other word), the Shannon (zero-error) capacity of the graph is given as
$$
\Theta(G) = \sup_{k} \sqrt[k]{\alpha(G^k)}.
$$
In spite of the importance of the Shannon capacity, remarkably few classes of graphs, such as perfect graphs \cite{Berge}, Kneser graphs and vertex-transitive self-complementary graphs \cite{Lovasz-0} and K\"onig-Egerv\'ary graphs \cite{Konig, Egervary, Konig-2, ind-num}, are known for which $\Theta(G)$ has been established analytically. In the majority of cases, these satisfy $\Theta(G) = \alpha(G)$, and are said to be class-1 graphs \cite{Berge}. In general, calculating $\Theta(G)$ is a very difficult problem and its value is not known even for graphs as simple as the seven cycle $C_7$! In his seminal paper \cite{Lovasz-0}, Lov\'asz introduced $\theta(G)$ as an efficiently computable  upper bound to the Shannon capacity: $\alpha(G) \leq \Theta(G) \leq \theta(G)$, using it to show that, for the five-cycle, $\Theta(C_5)=\sqrt{5}$.

Evidently, a simple way to prove the capacity of a graph is to show that $\alpha(G)=\theta(G)$, but we now see that this is only possible if $\omega_c=\omega_q$. With our characterisation of these cases and their corresponding graphs in hand, we are thus motivated to study the value of $\theta(G)$, giving a sufficient condition such that $\alpha(G)=\theta(G)$, thereby yielding a large, novel family of class-1 graphs.

{\em Non-local games and their graphs:} We consider the two-party {\sc xor} games in which both parties receive a question $x,y\in[m]$. Each gives a single bit of output $a,b\in\{0,1\}$, being tasked with winning with the maximum probability, where winning is defined as achieving the value $a\oplus b=f(x,y)$ for some binary function $f(x,y)$. Such a game is defined by the game matrix,
$$
\tilde{\Phi}=\sum_{x,y\in[m]}(-1)^{f(x,y)}P(x,y)\ket{x}\bra{y}.
$$
Of primary interest is the uniform probability distribution, in which the normalization factor is omitted:
$$
\Phi = \sum_{x,y \in [m]}(-1)^{f(x,y)}\ket{x}\bra{y}.
$$

Every {\sc xor} game has an associated a graph $G$ \cite{Chailloux, Winter, Fritz, LO}:
\begin{definition2}
The graph $G$ associated with the {\sc xor} game matrix $\Phi$ consists of $2m^2$ vertices $v\in V$. Each label $v$ can be expressed as $(x,y,a)$ where $x,y\in[m]$ and $a\in\{0,1\}$. Two vertices $v,v'\in V$ form an edge of the graph if ($x=x'$ and $a\neq a'$) or ($y=y'$ and $(-1)^{a\oplus a'} \neq \Phi_{xy}\Phi_{x'y}$).
\end{definition2}
This definition is equivalent to that of \cite{Chailloux}, having reduced the 4 labels $(x,y,a,b)$ used in \cite{Chailloux} via the winning relation $(-1)^{a\oplus b}=\Phi_{xy}$. The adjacency matrix of $G$ is conveniently expressed as
\begin{equation*}
\begin{split}
A(G)&=\identity\otimes(\proj{j}-\identity)\otimes X+\half\proj{j}\otimes\identity\otimes(\identity+X)	\\
&-\half [D(\proj{j}\otimes\identity)D]\otimes(\identity-X)
\end{split}
\end{equation*}
where $X$ is the usual Pauli-$X$ matrix and $\ket{j}$ is the all-ones vector $\ket{j}=\sum_{x\in[m]}\ket{x}$. $D$ is defined as
$$
D=\sum_{x,y\in[m]}\Phi_{x,y}\proj{x,y}.
$$
This graph is $(2m-1)$ regular, triangle free, and has a perfect matching \footnote{An example matching is all pairs of vertices $(x,y,0)$ and $(x,y,1)$.}. Its spectrum, and corresponding degeneracies, is readily found to be
\begin{equation}
\label{adj-spec}
\text{spec}(A(G))=\left\{\begin{array}{cc}
2m-1 & 1	\\
m-1 & 2m-2	\\
-1 & (m-1)^2 \\
1-m\pm\lambda_z & 1 \\
1 & m(m-2)
\end{array}\right.
\end{equation}
where $\lambda_z$ denotes the $m$ singular values of $\Phi$. %Generalisation to non-square matrices $\Phi$, or multi-player games, is straightforward but irrelevant to the scope of the present paper.

This spectrum was derived by first applying a Hadamard transform to the space of the third label, thereby splitting $A(G)$ into two subspaces $H_+\oplus H_-(D)$. Moreover, $H_+=-H_-(\identity)$. As such, it suffices to diagonalise $H_-(D)$. If $\ket{\lambda^A_z}$ and $\ket{\lambda^B_z}$ are the vectors of the singular value decomposition of $\Phi$ then observe that
\begin{equation}
\bra{\lambda^A_z}\bra{j}D\ket{j}\ket{\lambda^B_{z'}}=\lambda_z\delta_{z,z'},	\label{eqn:neat_trick}
\end{equation}
allowing one to verify that the only non-trivial eigenvectors of $H_-(D)$, with eigenvalues $1-m \mp\lambda_z$, are
$$
\ket{\eta_z^\pm}=\frac{\ket{\lambda^A_z}\ket{j}\pm D\ket{j}\ket{\lambda^B_z}}{\sqrt{2(m\pm\lambda_z)}}.
$$

{\em Games with no quantum advantage:} We are interested in categorising the {\sc xor} games which share the property of no quantum advantage. The optimal quantum strategy proceeds by Alice an Bob measuring $\pm 1$ observables $A_{x}$ and $B_{y}$ on a shared quantum state $\ket{\psi}$ when asked questions $x$ and $y$ respectively. This strategy can be represented in terms of unit vectors in $\mathbb{R}^{m}$ for each measurement of Alice ($\{|u_{x}\rangle \}$) and Bob ($\{|v_{y}\rangle\}$) \cite{Tsirelson, Wehner, Cleve}. The inner product $\langle u_{x} | v_{y} \rangle$ reproduces the expectation value of the measurement. For any such strategy the bias of the quantum value, $\epsilon_{q} \defeq 2 \omega_q - 1$, is given by
$
\epsilon_q = \text{Tr}[\tilde{\Phi}_{s} X]
$
where $X = \bigl(\begin{smallmatrix}A& S\\ S^T& B\end{smallmatrix} \bigr)$ and $\tilde{\Phi}_{s} = \left(\begin{smallmatrix}
0& \frac{1}{2} \tilde{\Phi} \\ \frac{1}{2} \tilde{\Phi}^T& 0
\end{smallmatrix} \right)$. $S$ is the strategy matrix and is defined as an $m \times m$ matrix having entries $S_{x,y} = \langle u_{x} | v_{y} \rangle$. The matrices $A, B$ with $A_{x, y} = \langle u_{x} | u_{y} \rangle$ and $B_{x,y} = \langle v_{x} | v_{y} \rangle$ describe local terms. The optimal quantum value is thus given by an optimisation over the vectors, which may be phrased as a semi-definite program $(P)$ \cite{Wehner, Cleve2}
\begin{eqnarray*}
\epsilon_q = \; \max \; \; &&\text{Tr[$\tilde{\Phi}_{s}$ X]}	 \\
\text{s.t.} \; \; &&\text{diag(X)} = |j \rangle, \; \; X \succeq 0,
\end{eqnarray*}
where $X \succeq 0$ denotes $X$ as positive semi-definite.
For a classical strategy, all vectors $|u_{x} \rangle$ and $|v_{y} \rangle$ are equal to $\pm|w\rangle$ for a single unit vector $|w\rangle$. The classical strategy matrix $S_c$ is thus a matrix with $\pm 1$ entries with all columns (and rows) being proportional to each other.

\begin{thm}
\label{xor-qeqc}
Consider a two-party {\sc xor} game with game matrix $\tilde{\Phi}$ with no all-zero row or column for which $S_c=\ket{s^A}\bra{s^B}$ represents the optimal classical strategy. Let $\Sigma = \text{diag}(\{\bra{i}\tilde{\Phi}\ket{s^B}\braket{s^A}{i}\}_{i=1}^m)$ and $\Lambda=\text{diag}(\{\bra{s^A}\tilde{\Phi}\ket{i}\braket{i}{s^B}\}_{i=1}^m)$.
There is no quantum advantage for $\tilde{\Phi}$ if and only if $\Sigma, \Lambda \succ 0$, or $\Sigma, \Lambda \prec 0$, and 
\begin{equation}
\label{nec-suff-1}
\rho(\Lambda^{-1} \tilde{\Phi}^T \Sigma^{-1} \tilde{\Phi}) = 1,
\end{equation}
where $\rho(.)$ denotes the spectral radius. 
\end{thm} 
\begin{proof}
The bias can be bounded from above by a feasible solution of the dual program $(D)$ to the semi-definite program given as \cite{Boyd, Cleve2}
\begin{eqnarray}
\label{Dual}
\min \; \; \sum_{i=1}^{2m}y_i \quad
\text{s.t} \quad \text{diag}(y) \succeq \tilde{\Phi}_s.
\end{eqnarray}
Notice that this problem is strongly dual \cite{Boyd}: the identity matrix is an explicit Slater point for $(P)$, and fixing $y$ to be the all-ones vector yields a Slater point for $(D)$. As such, we need to derive the conditions under which the solution to Eq.\ (\ref{Dual}) achieves the classical value, $\bra{s^B}\tilde{\Phi}\ket{s^A}$, which may also be written as $\bra{s^s}\tilde\Phi_s\ket{s^s}$ where $\bra{s^s}=({s^A}^T\;{s^B}^T)$. i.e.\ we require
$$
\Tr\left((\text{diag}(y)-\tilde\Phi_s)\proj{s^s}\right)=0.
$$
By the semi-definite condition of Eq.\ (\ref{Dual}), this means that $\ket{s^s}$ is a 0 eigenvector of $\text{diag}(y)-\tilde\Phi_s$:
\begin{equation}\label{eq:equality}
\text{diag}(y)\ket{s^s}=\tilde\Phi_s\ket{s^s}
\end{equation}
Simply making an element by element comparison, and remembering that $\ket{s^s}$ is a vector of $\pm 1$ entries, we have that whenever a classical strategy achieves the optimal quantum value, there is a unique optimal solution to the dual:
$
\text{diag}(y)=\bigl(\begin{smallmatrix}
\half\Sigma &0\\ 0& \half\Lambda
\end{smallmatrix} \bigr).
$ 
When can this be done? The constraint in Eq. (\ref{Dual}) can be rewritten as
 $\bigl(\begin{smallmatrix}
\Sigma & -\tilde{\Phi}\\ -\tilde{\Phi}^T& \Lambda
\end{smallmatrix} \bigr) \succeq 0$. Theorem 7.7.7 in \cite{Horn} shows that when $\tilde{\Phi}$ has non-zero rows or columns, this is equivalent to $\Sigma, \Lambda \succ 0$ or $\Sigma, \Lambda \prec 0$, and 
\begin{equation}
\rho(\tilde{\Phi}^T \Sigma^{-1} \tilde{\Phi} \Lambda^{-1}) \leq 1.
\end{equation}
Since the solution to the dual giving the classical bound must saturate this inequality (Eq.\ (\ref{eq:equality})), it can be replaced by equality.
\end{proof}

When $S_c = S_c^T$ and $\tilde{\Phi} = \tilde{\Phi}^T$, the condition of Thm.\ \ref{xor-qeqc} reduces to $\pm \Sigma \succ 0$ and $\rho(\Sigma^{-1} \tilde{\Phi}) = 1$.

\begin{corollary} \label{cor:1}
If the vectors corresponding to the maximum singular value of $\tilde{\Phi}$ only contain elements that are $\pm 1$, then there is no quantum advantage for players of the game $\tilde{\Phi}$.
\end{corollary}
\begin{proof}
Let the (unnormalised) maximum singular vectors be $\ket{\lambda^A}$ and $\ket{\lambda^B}$ such that $\bra{\lambda^A}\tilde{\Phi}\ket{\lambda^B}=\lambda$ is the maximum singular value. In this case, $S_C=\ket{\lambda^A}\bra{\lambda^B}$. Then $\tilde{\Phi}S_c^T=\lambda\proj{\lambda^A}$ such that $\Sigma=\Lambda=\lambda\identity$. Evidently, these are positive and
$$
\rho(\tilde{\Phi}^T \Sigma^{-1} \tilde{\Phi} \Lambda^{-1})=\frac{1}{\lambda^2}\rho(\tilde{\Phi}^T\tilde{\Phi})=1.
$$
\end{proof} 
This is only a sufficient condition and not necessary as one. For example, the maximum eigenvector of
\begin{equation}	\label{eqn:ex}
\tilde{\Phi}_{ex} = \frac{1}{16} \left(\begin{matrix}
1 & -1 & -1 & 1 \\ -1 & -1 & 1 & -1 \\ -1 & 1 & -1 & -1 \\ 1 & -1 & -1 & 1
\end{matrix} \right)
\end{equation} 
does not consist of $\pm 1$ elements, and yet it can be verified via Thm.\ \ref{xor-qeqc} that $\omega_q(\tilde{\Phi}_{ex}) = \omega_{c}(\tilde{\Phi}_{ex})$.

We now have a simple way to construct games with no quantum advantage. The NLC games are a trivial case since $\tilde{\Phi}$ is diagonal in the Hadamard basis in that instance \cite{NLC}. To construct novel examples, we start with two observations: (i) If $\tilde\Phi^1$ and $\tilde\Phi^2$ are two game matrices satisfying Cor.\ \ref{cor:1}, then it follows that$\tilde\Phi^1\otimes\tilde\Phi^2$ also satsifies Cor.\ \ref{cor:1} (see also \cite{Cleve2}), (ii) If $\tilde\Phi$ is a game matrix satisfying Cor.\ \ref{cor:1}, then so does $U\tilde\Phi V^T$ where $U,V$ are arbitrary orthogonal transformations that map the maximum singular vectors onto other vectors with $\pm 1$ entries. The first observation permits us to extend any examples for small input size to examples with arbitrary size of input.

In the uniform probability case, it suffices to construct any symmetric matrix $\Phi\in\{\pm 1\}^{m\times m}$ for which the total of every row is the same, and at least $\half m$, which ensures that the maximum eigenvector is $\ket{j}$. Since this family of examples does not readily extend to general probability distributions, it may be less interesting for information theoretic purposes. Instead, we address a question of \cite{NLC}: finding {\sc xor} games that differ from NLC, but with no quantum advantage for fixed patterns of input probability distribution. Consider the anti-circulant matrices: for any $m$, $\ket{j}$ is an eigenvector, and if $m$ is even, so is the alternating signs vector. All we have to do is further restrict the matrix elements to guarantee that one of these yields the eigenvalue of maximum modulus. In the case of $m=4$, we specify $\tilde\Phi$ by giving its first row: $(\gamma_0,\gamma_1,\gamma_2,\gamma_3)$ subject to $\sum_i|\gamma_i|=\frac{1}{4}$. Provided
$$
\max\left(\left(\sum_{i=0}^3\gamma_i\right)^2,\left(\sum_{i=0}^3\gamma_i(-1)^i\right)^2\right)\geq (\gamma_0-\gamma_2)^2+(\gamma_1-\gamma_3)^2,
$$
we have a game for which there is no quantum advantage. A sufficient condition for this to happen is $\gamma_0\gamma_2+\gamma_1\gamma_3\geq 0$. Many different patterns for the probability distribution, such as 
$$
\tilde\Phi=\left(\begin{array}{cccc}
p & q & q & -p \\
q & q & -p & p \\
q & -p & p & q \\
-p & p & q & q
\end{array}\right),
$$
where $|p|+|q|=1/8$ and $p,q\in\mathbb{R}$, satisfy the condition, thereby giving game matrices for which there is never a quantum advantage.

{\em Shannon capacity from game graphs:} Game graphs for which $\omega_c=\omega_q$ are good candidates for those that might have a Shannon capacity $\Theta(G)=\alpha(G)$.

\begin{thm}
Every two-party {\sc xor} game with $m$ uniformly chosen inputs for each party, and satisfying Cor.\ \ref{cor:1} has a game graph which is class-1 (has $\Theta(G) = \alpha(G)$).
\end{thm}
\begin{proof}
To establish the Shannon capacity, our strategy is to find both $\alpha(G)$ and $\theta(G)$. $\alpha(G)$ is straightforward -- it coincides with the optimal strategy, being specified by the maximum singular vectors. Thus,
$$
\alpha(G)=\half m(m+\|\Phi\|).
$$
In order to compute $\theta(G)$, we use the following characterization of the Lov\'asz theta number derived in \cite{Lovasz-0}
\begin{thm}[\cite{Lovasz-0}]
\label{lovasz-upp-bound}
Let $G$ be a graph on vertices $\{1, \dots, N\}$. Then $\theta(G)$ is the minimum of the largest eigenvalue of any symmetric matrix $(\tilde A_{i,j})_{i,j=1}^{N}$  such that
\begin{equation}
\tilde A_{i,j} = 1, \quad \text{if i = j or if i and j are non-adjacent}.
\end{equation}
\end{thm}
\noindent
Indeed, it is sufficient for our purposes to find any symmetric matrix $\tilde{A}$ whose maximum eigenvalue matches $\alpha(G)$ since $\alpha(G)\leq\theta(G)$. Define the matrix
$$
\tilde A \defeq \proj{j}\otimes\proj{j}\otimes(\identity+X)+aA(G)+b\identity \otimes\identity \otimes X.
$$
It is readily verified that $\tilde A$ satisfies the conditions of Theorem \ref{lovasz-upp-bound}. All of the terms in $\tilde A$ commute with each other, so the diagonalisation is the same as for $A(G)$, and the eigenvalues are readily extracted. It is now our task to select $a,b$ such that the largest eigenvalue is as small as possible. We clearly require $a<0$, in which case there are 3 relevant eigenvalues: $2m^2+b+a(2m-1)$, $b-a$, and $-b+a(1-m-\|\Phi\|)$. The minimum arises when all three are equal: $a=-m$ and $b=\alpha(G)-m$, yielding a maximum eigenvalue of $\alpha(G)$. We conclude that $\alpha(G)=\theta(G)$, and the graph is class-1.
\end{proof}
This proof automatically covers all NLC games, but also includes many other XOR games. A further consequence is that whenever $\omega_q=\half (1+ \frac{1}{m^2}\|\Phi\|)$, we know that $m^2 \omega_q=\theta(G)$. The well-known CHSH game \cite{chsh} is an example of this.

The family of game graphs described here is distinct from previously described families of class-1 graphs: (i) if a graph belongs to the K\"onig-Egerv\'ary (KE) family, it has $\alpha(G) + \nu(G)$ vertices, where $\nu(G)$ denotes the maximum size of a matching. Since the {\sc xor} game graphs have a perfect matching they can only belong to the KE family in the trivial case of $\omega_c = 1$. (ii) For perfect graphs, the classical (non-contextual) polytope coincides with the general consistent polytope \cite{Cabello2, GLS, Cabello}, so these again correspond to a trivial case where there is no advantage not only in quantum theory but also in general no-signalling theories. (iii) For Kneser graphs on $n$ vertices \cite{Lovasz-0}, it is known that
$$
\theta(G) = \frac{- |V| \lambda_{\min}}{\lambda_{\max} - \lambda_{\min}} 
$$
where $\lambda_{\max/\min}$ are the corresponding maximum and minimum eigenvalues of the adjacency matrix. From Eq.(\ref{adj-spec}), we see that this only happens for {\sc xor} games if $\| \Phi\| = m$, where again $\omega_c = 1$.

{\em Conclusions:} We have given a necessary and sufficient condition under which a bipartite {\sc xor} game gives no quantum advantage. This yields broad new families of games for which there is never a quantum advantage, independent of the underlying probability distribution. However, those that we have generated all rely on ensuring that the optimal classical strategy coincides with the vectors of the maximum singular value of the game matrix. It would be particularly interesting from an information theoretic standpoint as to whether there exist any such classes which do not require that coincidence.

Motivated by this broad range of games that can be designed, we then showed that the associated game graphs are all class-1, i.e.\ their Shannon capacity coincides with the independence number of the graph. This result, remarkable simply due to the difficulty in evaluating the Shannon capacity even for very simple graphs, is an entirely classical result derived as a consequence of insight provided from the study of quantum mechanical problems. The proof of this required Cor.\ \ref{cor:1} to hold. However, we conjecture that this restriction can be lifted, and that a necessary and sufficient condition for game graphs to have $\alpha(G)=\theta(G)$ is given by Thm.\ \ref{xor-qeqc} (necessity is trivial since $\omega_q\leq\theta(G)$). For example, it can be verified that this is true for the example in Eq.\ (\ref{eqn:ex}).

{\em Acknowledgements.} This work is supported by the ERC AdG grant QOLAPS and also forms part of the Foundation for Polish Science TEAM project co-financed by the EU European Regional Development Fund. G.M. acknowledges support from the grant IDEAS PLUS and the Brazilian agency Fapemig. We thank S. Severini, D. Roberson and M. Horodecki for useful discussions and A.\ Varvitsiotis for feedback on the manuscript.

\end{document}